\documentclass[
pra,reprint,superscriptaddress,nofootinbib]{revtex4-2}

\usepackage{times}
\usepackage{fancyhdr}
\usepackage{graphicx}
\usepackage{float}
\usepackage{tabularx,booktabs} 
\usepackage{enumitem}

\usepackage{physics}
\usepackage{amsmath}
\usepackage{amssymb}
\usepackage{dsfont}
\usepackage{mathrsfs} 
\usepackage{euscript}
\usepackage{upgreek}
\usepackage{mathtools}
\usepackage{mathdots}
\usepackage[bbgreekl]{mathbbol}
\DeclareSymbolFontAlphabet{\mathbbol}{bbold}
\DeclareSymbolFontAlphabet{\mathbb}{AMSb}

\usepackage{amsmath}
\usepackage{amsthm}
\theoremstyle{remark}
\newtheorem{theorem}{Theorem} 
\newtheorem{lemma}{Lemma}

\newtheorem{corollary}{Corollary} 
\newtheorem{remark}{Remark} 

\usepackage[usenames,dvipsnames,table]{xcolor}
\definecolor{blue_ref}{RGB}{46,48,146}

\usepackage[colorlinks=true,urlcolor=blue_ref,citecolor=blue_ref,linkcolor=blue_ref]{hyperref}

\usepackage{comment} 
\usepackage{soul} 
\setstcolor{red}
\usepackage{cancel}
\usepackage{lipsum}

\usepackage{pgffor}

\usepackage{pdfpages} 
\makeatletter
\AtBeginDocument{\let\LS@rot\@undefined}
\makeatother

\usepackage[us,hhmmss]{datetime} 

\newcommand{\m}[1]{\mathit{#1}}
\newcommand{\f}[1]{\mathrm{#1}}
\newcommand{\map}[1]{\ifcat\noexpand#1\relax#1\else{\mathcal{#1}}\fi}
\newcommand{\set}[1]{\mathbbol{#1}}

\newcommand{\onorm}[1]{{\left\vert\kern-0.25ex\left\vert\kern-0.25ex\left\vert #1 
    \right\vert\kern-0.25ex\right\vert\kern-0.25ex\right\vert}}
    
\newcommand{\eq}[2]{\begin{equation} \label{eq:#1} #2 \end{equation}}
\newcommand{\alsub}[2]{\begin{subequations}\label{eq:#1}\begin{align} #2 \end{align}\end{subequations}}
\let\leq\leqslant
\let\geq\geqslant

\newcommand{\qs}{S}
\newcommand{\id}[1]{\m{I}_{#1}}
\newcommand{\hil}{\mathscr{H}}
\newcommand{\los}{\mathfrak{L}}
\newcommand{\trc}{\mathrm{tr}}
\newcommand{\prob}{\mathbb{P}}
\newcommand{\expv}{\mathbb{E}}
\newcommand{\vari}{\mathbb{V}}
\newcommand{\cov}{\mathbb{C}\mathrm{ov}}
\newcommand{\cor}{\mathbb{C}\mathrm{or}}
\newcommand{\pcc}{\mathbb{R}}
\newcommand{\ent}{\mathbb{H}}
\newcommand{\mi}{\mathbb{I}}

\newcommand{\braf}[1]{\mbox{$\langle #1 |$}}
\newcommand{\ketf}[1]{\mbox{$| #1 \rangle$}}
\newcommand{\braketf}[2]{\mbox{$\langle #1 | #2 \rangle$}}

\def\HarvardSEAS{John A. Paulson School of Engineering and Applied Sciences, Harvard University, Cambridge, Massachusetts 02138, USA}
\def\ANU{Centre for Quantum Computation and Communication Technology,
Department of Quantum Science, Australian National University, Canberra, ACT 2601, Australia}
\def\ASTAR{A*STAR Quantum Innovation Centre (Q.InC), Institute of Materials Research and Engineering (IMRE), Agency for Science, Technology and Research (A*STAR), 2 Fusionopolis Way \#08-03 Innovis, 138634, Singapore}
\def\UCLA{Physical Sciences, College of Letters and Science, University of California, Los Angeles (UCLA), CA, USA}

\begin{document}

\title{Quantifying total correlations in quantum systems through the Pearson correlation coefficient}

\author{Spyros Tserkis}
\email{spyrostserkis@gmail.com}
\affiliation{\HarvardSEAS}
\affiliation{\UCLA}
\author{Syed M. Assad}
\affiliation{\ANU}
\affiliation{\ASTAR}
\author{Ping Koy Lam}
\affiliation{\ASTAR}
\affiliation{\ANU}
\author{Prineha Narang}
\email{prineha@ucla.edu}
\affiliation{\HarvardSEAS}
\affiliation{\UCLA}


\begin{abstract}
Conventionally the total correlations within a quantum system are quantified through distance-based expressions such as the relative entropy or the square-norm. Those expressions imply that a quantum state can contain both classical and quantum correlations. In this work, we provide an alternative method to quantify the total correlations through the Pearson correlation coefficient. Using this method, we argue that a quantum state can be correlated in either a classical or a quantum way, i.e., the two cases are mutually exclusive. We also illustrate that, at least for the case of two-qubit systems, the distribution of the correlations among certain locally incompatible pairs of observables provides insight in regards to whether a system contains classical or quantum correlations. Finally, we show how correlations in quantum systems are connected to the general entropic uncertainty principle.
\end{abstract}

\maketitle

\section{Introduction}

The correlation between two random variables in classical probability theory is typically measured via either the Pearson correlation coefficient (PCC) \cite{Pearson_Galton_PRSL_95, Renyi_B_07} or the mutual information (MI) \cite{Cover_Thomas_B_06}. Given two probability distributions for two random variables, the PCC expression is based on the first and second statistical moments of the distributions, whilst the MI expression is based on the Kullback--Leibler divergence between the joint and marginal distributions. The two measures can be thought of as reciprocal to each other.

In quantum systems the quantification of correlations is predominately based on measures that generalize or modify the concept of MI, e.g., for the quantification of entanglement the Kullback--Leibler divergence is substituted by the notion of relative entropy~\cite{Umegaki_KMJ_62, Vedral_RMP_02}. On the other hand, PCC has been only recently used in the quantum context in order to witness entanglement~\cite{Maccone_Bruss_Macchiavello_PRL_15, Tserkis_etal_PLA_24} and to construct Bell inequalities~\cite{Pozsgay_etal_PRA_17, Huang_Vontobel_IEEE_21, Tserkis_etal_PLA_24}.

In this work, we introduce a PCC-based expression that quantifies the total correlations in a quantum system, generalizing previous works that were based on the covariance~\cite{Kothe_Bjork_PRA_07, Abascal_Bjork_PRA_07}. This method complements the existing ones by offering an alternative way to quantify correlations. The main advantage of this new method is the insight it provides in regards to the type of correlations existent in the quantum system based on the distribution of correlations among certain pairs of observables. Finally, the role of the uncertainty principle in regards to the type of correlations in quantum systems is revealed.

The paper is structured as follows. Section~\ref{sec:preliminaries} provides preliminary information about quantum systems, the notion of complementarity, and the different ways that correlations can be measured. In Section~\ref{sec:total_correlations} the new method for quantifying the total correlations in a quantum system is introduced, followed by Section~\ref{sec:total_correlations_qubits}, where two-qubit states are studied. A discussion of the results is given in Section~\ref{sec:discussion}. Finally, the work is concluded in Section~\ref{sec:conclusion}.

\section{Preliminaries}
\label{sec:preliminaries}

\subsection{Quantum Systems}

A pure state is represented by a normalized vector $\ket{\psi}$ in the Hilbert space $\hil_d$, with $d \in \set{N} : d \geq 2$. A mixed state is represented by a positive unit-trace operator $\qs$ that belongs in the set of all linear and bounded operators $\los(\hil_d)$. A necessary and sufficient condition for a state to be pure is to satisfy the equation $\tr ( \qs^2) = 1$. All states admit a spectral decomposition, $\qs = \sum_{k=0}^{d-1} p_k \ketbra{\psi_k}{\psi_k}$. An observable, associated with a physical property, is represented by a Hermitian operator, $\m{X} = \sum_{k=0}^{d-1} x_k \ketbra{x_k}{x_k} \in \los(\hil_d)$, where $\{ \ket{x_k} \}_{k=0}^{d-1}$ is an orthonormal basis in $\hil_d$. The eigenspectrum $\{ x_k \}_{k=0}^{d-1}$ corresponds to the outcomes of the rank-1 projective measurements $\{ \ketbra{x_k}{x_k} \}_{k=0}^{d-1}$, occurring with probabilities $\prob_{\qs}(x_k) \coloneqq \bra{x_k} \qs \ket{x_k}$. For an operator $\m{X}$, sample space $\set{\Omega}_{\m{X}}$ is the set of mutually exclusive and collectively exhaustive measurement outcomes. The Shannon entropy of an observable $\m{X}$ is defined as $\ent_{\qs}(\m{X}) \coloneqq - \sum_{k=0}^{d-1} \prob_{\qs}(x_k) \log_2 \prob_{\qs}(x_k)$.

\subsection{Complementarity}

\begin{figure}[t!]
\centering
\includegraphics[width=\columnwidth]{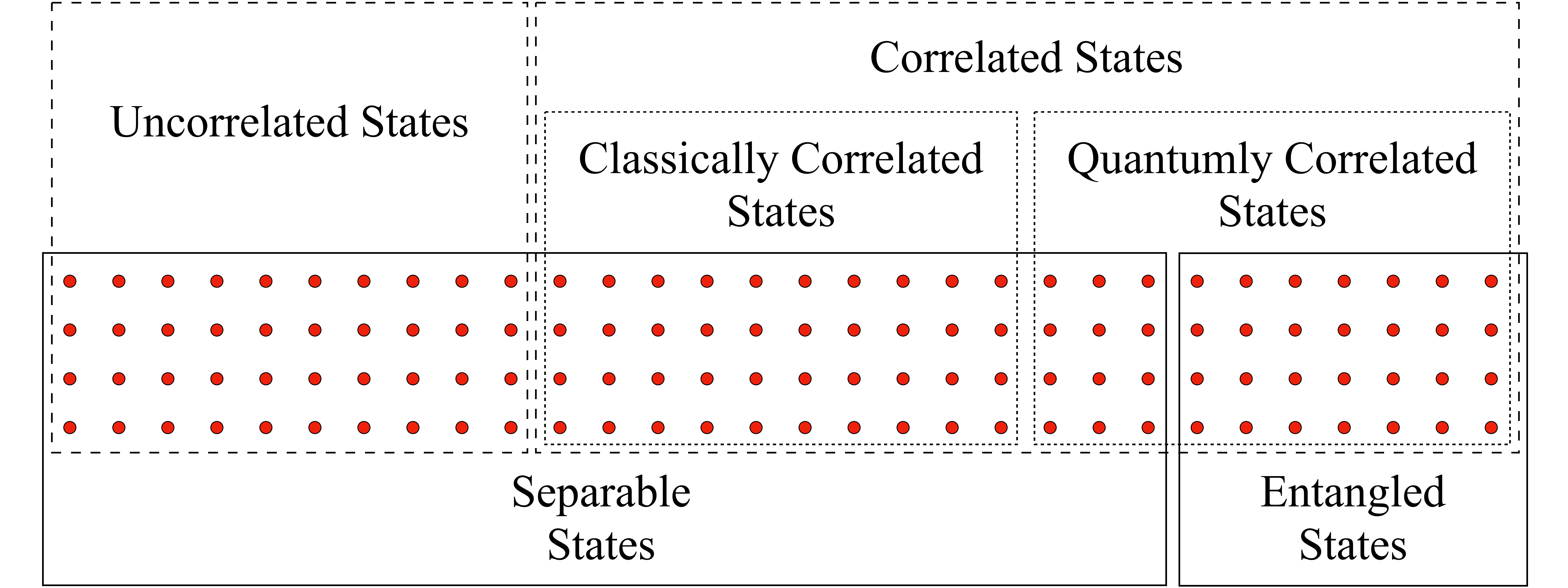}
\caption{Venn diagram showing how quantum states (represented for simplicity by a finite number of dots) can be classified according to the correlations existing in them. Every state is either correlated or uncorrelated. Correlated states are further divided into classically correlated and quantumly correlated states. A different type of classification for a state is that it is either separable or entangled.}
\label{fig:correlations}
\end{figure}

The general entropic uncertainty principle between two non-degenerate observables $\m{X}$ and $\m{Y}$ that act on the same Hilbert space $\los(\hil_d)$ is expressed by the following relation~\cite{Kraus_PRD_87, Maassen_Uffink_PRL_88}:
\eq{uncertainty_principle}{
\ent_{\qs}(\m{X})+ \ent_{\qs} (\m{Y}) \geq - \log \Big( \max\limits_{k, \ell} |\braket{x_k}{y_{\ell}}|^2 \Big) \,.
}
The term $-\log \big(\max_{k, \ell} |\braket{x_k}{y_{\ell}}|^2 \big) \in [0, \log d ]$ quantifies the incompatibility between the two observables, and the uncertainty principle asserts that the more certainty we have for the outcome of one observable the more uncertainty we have for the outcomes of the other, and vice versa (for other measures of incompatibility between the two observables see Refs.~\cite{Heinosaari_Miyadera_Ziman_JPA_16, Guhne_etal_RMP_23}). Systems that contain only compatible observables are considered classical, while quantum systems can contain both compatible and incompatible observables. Two maximally incompatible observables are called complementary~\cite{Schwinger_PNAS_60}, and have mutually unbiased bases (MUB), i.e.,
\eq{complementarity_relation}{
|\braketf{x_k}{y_{\ell}}|^2=\frac{1}{d}  \quad \forall \, k, \ell \in \{ 0, \cdots, d-1 \} \,.
}

The maximum number of pairwise complementary observables in $\los(\hil_d)$ is equal to the maximum number of MUB in $\hil_{d}$. For a Hilbert space $\hil_{d}$, complete is called a set of MUB when it has $d+1$ elements. It is an open question whether a complete set of MUB exists for arbitrary dimensions~\cite{Horodecki_Rudnicki_Zyczkowski_PRXQ_22}, however it is known that the maximum number of MUB for any Hilbert space cannot exceed $d+1$, and that a complete set can always be found when $d$ is a power of a prime number ~\cite{Ivonovic_JPAMG_81, Wootters_Fields_AN_89, Bandyopadhyay_etal_A_02}. For example, when $d=2$ the Pauli matrices constitute a complete set of three pairwise complementary observables. For complementary observables in higher dimensions see Refs.~\cite{Lawrence_PRA_04, Brierley_Weigert_Bengtsson_QIC_10}. The significance of MUB can be seen by the fact that measurements on them provide enough information to fully reconstruct a state~\cite{Adamson_Steinberg_PRL_10, Fernandez_Saavedra_PRA_11, Petz_Laszlo_RMP_12}.

\subsection{Correlations in Quantum Systems}

Any bipartite state $\qs \in \los(\hil_n) \otimes \los(\hil_m)$ can be written in the Fano form~\cite{Fano_RMP_83, Schlienz_Mahler_PRA_95} as follows:
\eq{fano_form}{
\qs = \frac{\id{} {\otimes} \id{} + \sum\limits_{k=1}^{d^2-1} n_k \m{\Lambda}_k {\otimes} \id{} + \sum\limits_{\ell=1}^{d^2-1} s_{\ell} \id{} {\otimes} \m{\Lambda}_{\ell} + \sum\limits_{k, \ell=1}^{d^2-1} t_{k \ell} \m{\Lambda}_k {\otimes} \m{\Lambda}_{\ell}}{d^2} \,,
}
where $d = \min \{ n,m\}$, $\id{}$ the identity operator, and $\{ \m{\Lambda}_k \}_{k=1}^{d^2-1}$ is a Hermitian operator basis with $\trc(\m{\Lambda}_k)=0$ and $\trc(\m{\Lambda}_k \m{\Lambda}_{\ell} ) = d \delta_{k \ell}$. The coefficients $n_k$, $s_{\ell}$, and $t_{k \ell}$ are expressed as follows:
\alsub{n_s_t_coeff}{
n_k &= \tr[\qs (\m{\Lambda}_k \otimes \id{}) ] \,, \\
s_{\ell} &= \tr[\qs (\id{} \otimes \m{\Lambda}_{\ell}) ] \,, \\
t_{k \ell} &= \tr[\qs (\m{\Lambda}_k \otimes \m{\Lambda}_{\ell}) ] \,.
}

Quantifying correlations in quantum systems is a challenging task given that, unlike in classical systems, there are different types of correlations~\cite{Groisman_Popescu_Winter_PRA_05, Li_Luo_PRA_07, Adesso_Bromley_Cianciaruso_JPA_16, Modi_etal_RMP_12}. More specifically, there are two types of bipartite states: (i) uncorrelated states, i.e.,
\eq{}{
\qs_{\text{prod}} = \qs^{(\text{A})} \otimes \qs^{(\text{B})} \,\, \text{with} \,\, \tr \bigl( \qs_{\text{prod}}^2 \bigl) < 1 \,,
}
and (ii) correlated states.

\begin{remark}
The concept of correlation is inapplicable to pure product states given that each side is by construction fixed and non-probabilistic. In practice, a value for pure product states for a given correlation measure can be assigned by considering an appropriate limit.
\end{remark}

Correlated states can be further classified as either: (i) classically correlated~\cite{Henderson_Vedral_JPAMG_01, Ollivier_Zurek_PRL_01, Oppenheim_etal_PRL_02}, that have the form:
\eq{classically_correlated}{
\qs_{\text{cl}} \coloneqq \sum_{k, \ell} p_{k \ell} \ketbra{e_k}{e_k} \otimes \ketbra{h_{\ell}}{h_{\ell}} \,,
}
where $\sum_{k, \ell} p_{k \ell} =1$, and $\{ \ket{e_k} \}$, $\{ \ket{h_{\ell}} \}$ being orthonormal bases, or (ii) quantumly correlated\footnote{Note that the clear dichotomy between classically and quantumly correlated states is not a distinction that all authors accept. For example, as we discuss in the section \ref{sec:mutual_information}, an alternative view is that a state can be partially classically correlated and partially quantumly correlated.}.

Uncorrelated and classically correlated states can be written as a convex combination of product states, i.e.,
\eq{separable_states}{
\qs_{\text{sep}} \coloneqq \sum_k p_k \qs_k^{(\text{A})} \otimes \qs_k^{(\text{B})} \,,
}
while this is not always true for quantumly correlated states. States that can be written in the form of Eq.~\eqref{eq:separable_states} are called separable, otherwise they are called entangled~\cite{Werner_PRL_89} (see Refs.~\cite{Horodecki_etal_RMP_09, Plenio_Virmani_B_14} for a review of entanglement theory). In Fig.~\ref{fig:correlations} we schematically represent all possible types of correlations in quantum systems.

Below, we introduce the two main ways correlations can be quantified in quantum systems.

\subsubsection{Mutual Information} \label{sec:mutual_information}

One way to quantify the different types of correlations for a given bipartite state $\qs$ is by invoking a distant measure, e.g., relative entropy~\cite{Umegaki_KMJ_62, Vedral_RMP_02} which for two states $\qs_k$ and $\qs_{\ell}$ is defined as
\eq{relative_entropy}{
\ent ( \qs_k \| \qs_{\ell}) \coloneqq \trc[\qs_k (\log_2 \qs_k - \log_2 \qs_{\ell})] \,,
}
and minimize it over a certain class of states\footnote{An alternative distant measure is the square-norm, defined as $\mathbb{D}(\qs_k,\qs_{\ell}) \coloneqq \norm{ \qs_k - \qs_{\ell} }_{\text{tr}}^2$~\cite{Dakic_Vedral_Brukner_PRL_10, Bellomo_etal_PRA_12}.}.

Based on this approach, quantum correlations, viz., discord, are measured as the minimum relative entropy of the state over all classically correlated states, $\f{Q}(\qs) \coloneqq \min_{\qs_\text{cl}} \ent ( \qs \| \qs_\text{cl})$,~\cite{Modi_etal_PRL_10}, and entanglement as the minimum relative entropy of the state over all separable states, viz., relative entropy of entanglement, $\f{E}(\qs) \coloneqq \min_{\qs_\text{sep}} \ent ( \qs \| \qs_\text{sep})$,~\cite{Vedral_etal_PRL_97}. The total correlations are measured as the minimum distance of the state over all uncorrelated states, which is equal to the mutual information $\mi_{\qs} (\text{A} : \text{B})$~\cite{Modi_etal_PRL_10}, i.e.,
\eq{mutual_information}{
\f{T}(\qs) \coloneqq \min_{\qs_\text{prod}} \ent ( \qs \| \qs_\text{prod}) = \ent ( \qs \| \qs^{(\text{A})} \otimes \qs^{(\text{B})}) = \mi_{\qs} (\text{A} : \text{B}) \,.
}
Classical correlations are defined as $\f{C}(\qs) \coloneqq \ent ( \qs \| \tilde{\qs}_{\text{prod}}) - \ent ( \qs \| \tilde{\qs}_{\text{cl}})$, where $\tilde{\qs}_{\text{cl}}$ is the closest classical state to $\qs$, and $\tilde{\qs}_{\text{prod}}$ is the closest product state to $\tilde{\qs}_{\text{cl}}$. Note that for pure product states the values of $\f{C}(\cdot)$, $\f{Q}(\cdot)$, and $\f{T}(\cdot)$ converge to 0 due to the limit $\lim\limits_{x \rightarrow 0} x \log_2 x = 0$.

Using this approach the total correlations of a state can be computed analytically, but this is not the case for classical correlations, quantum correlations, and entanglement (relative entropy of entanglement), for which in general no simple analytical expression is currently known. Note that this method implies that a quantum state can be partially correlated in a classical way and partially correlated in a quantum way, an idea that leads to the following issue.

\begin{remark}
The total correlations measured through $\f{T}(\qs)$ are not in general equal to the sum of the classical and quantum correlations: $\f{C}(\qs) + \f{Q}(\qs)$~\cite{Modi_etal_PRL_10}. An extra quantity to offset this imbalance can be introduced which however lacks a physical meaning.
\end{remark}

\subsubsection{Pearson Correlation Coefficient}

A different way to quantify the amount of correlations within a bipartite state is the Pearson correlation coefficient (PCC)~\cite{Maccone_Bruss_Macchiavello_PRL_15}. Contrary to the geometrical method, which is calculated through the quantum states that comprise a quantum system, the PCC is directly calculated through pairs of observables that can be measured in a quantum system. In particular, for two observables that act on two different Hilbert spaces, $\m{A} \in \los(\hil_d)$ and $\m{B} \in \los(\hil_d)$, that have sample spaces of equal cardinality, the PCC is defined as
\eq{pcc_definition}{ 
\cor_{\qs}(\m{A},\m{B}) \coloneqq  \frac{\cov_{\qs}(\m{A},\m{B})}{\sqrt{\vari_{\qs}(\m{A} ) \vari_{\qs}(\m{B})}} \,,
}
with $\cov_{\qs}(\m{A},\m{B}) \coloneqq \expv_{\qs}(\m{A} \otimes \m{B})  - \expv_{\qs}(\m{A}) \expv_{\qs}(\m{B})$ denoting the covariance between $\m{A}$ and $\m{B}$. $ \expv_{\qs}(\m{A})$, $ \expv_{\qs}(\m{B})$, and $ \expv_{\qs}(\m{A} \otimes \m{B})$ are the expectation values of $\m{A}$, $\m{B}$, and $\m{A} \otimes \m{B}$, respectively, and finally $\vari_{\qs}(\m{A} )$ and $\vari_{\qs}(\m{B})$ the variances of $\m{A}$ and $\m{B}$, respectively. For an arbitrary observable $\m{X}$ the expectation value is defined as $\expv_{\qs}(\m{X}) \coloneqq \trc( \qs \m{X})$ and the variance as $\vari_{\qs}(\m{X}) \coloneqq \expv_{\qs}(\m{X}^2) -\expv_{\qs}(\m{X})^2$. The range of the PCC lies within $[-1,1]$, where $+1$ implies perfect correlation and $-1$ perfect anti-correlation. For independent subsystems the PCC is zero for any pair of observables. For vanishing variances the PCC is undefined, e.g., the PCC for pure product states can lead to a singularity $0/0$.

In both classical and quantum regimes the PCC quantifies only linear relationships between random variables and observables, respectively. General monotonic relationships are quantified by the Spearman rank correlation coefficient~\cite{Spearman_BJP_06}, which is defined classically as the PCC between random variables with ranked values, and in the quantum context we can immediately generalize it as the PCC between observables with ranked eigenspectra, i.e., eigenspectra that can be ranked in an ascending or a descending order. Note that for two dimensional observables the PCC coincides with the Spearman rank correlation coefficient.

\section{Total Correlations Through PCC}
\label{sec:total_correlations}

\begin{figure}[t]
\centering
\includegraphics[width=0.8\columnwidth]{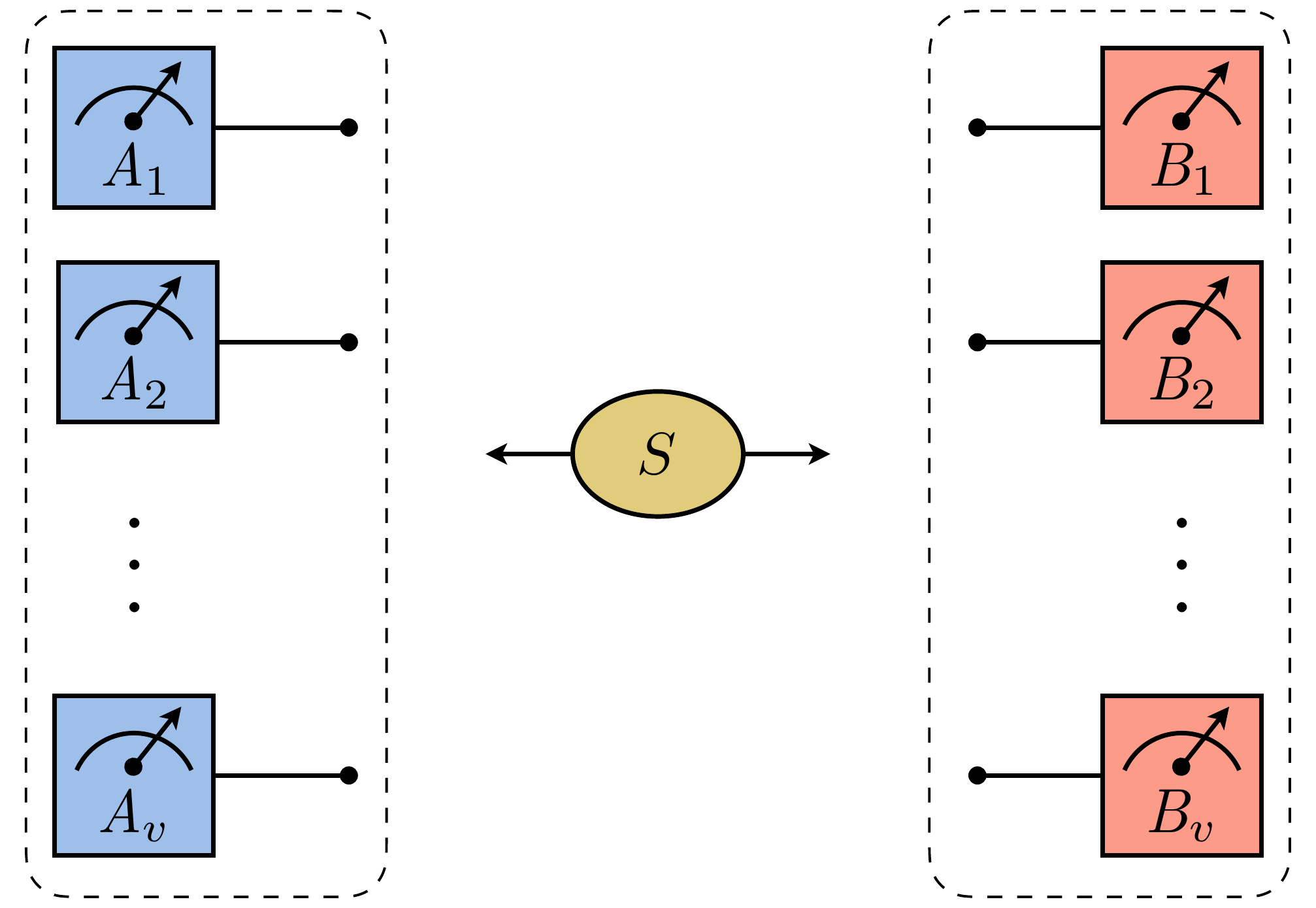}
\caption{A bipartite states $\qs$ measured in one Hilbert space by an observable $\m{A}_i$, with $k \in \{1, 2, \cdots, v \}$ and in the other Hilbert space by an observable $\m{B}_i$, with $k \in \{1, 2, \cdots, v \}$. Within each Hilbert space the observables are complementary.}
\label{fig:twosides}
\end{figure}

In this section we propose a method to quantify the total correlations in a quantum system based on the PCC. For an arbitrary bipartite state there are in general multiple pairs of observables that can be potentially correlated with each other. In order to capture all the available information in the state with the minimum amount of measurements, we consider the maximum set of complementary observables in the Hilbert space of each side as seen in Fig.~\ref{fig:twosides}. Also, in order to capture any possible type of correlations, i.e., not only linear, we consider observables with ordered eigenspectra. Thus, for the ordered observables $\{ (\m{A}_{i},\m{B}_{i}) \}_{i=1}^v$, where $v$ is the cardinality of the maximum set of pairwise complementary observables in the Hilbert space $\los(\hil_d)$, the sum of the correlations of a bipartite state $\qs$ is
\eq{sum_correlations}{ 
\sum_{i=1}^{v} \left| \cor_{\qs}(\m{A}_{i},\m{B}_{i}) \right| \,.
}
In Ref.~\cite{Sadana_etal_QIP_24} it was shown that for a specific set of observables and for certain families of states a sum of the PCC values is connected to the entanglement measure Negativity~\cite{Zyczkowski_etal_PRA_98, Lee_etal_JMO_00, Vidal_Werner_PRA_02}.

Even though it is reasonable for correlations to depend on the measured observables, in quantum information theory, expressions that quantify correlations are typically defined in such a way so they are invariant under local unitary transformations. Thus, we define the measure for total correlations as follows
\eq{total_correlations_measure}{
\pcc_{\qs}(\text{A} : \text{B}) \coloneqq \max_{\{\m{A}_i,\m{B}_i \}}   \sum_{i=1}^{v} \left| \cor_{\qs}(\m{A}_{i},\m{B}_{i}) \right| \,,
}
where the maximization is taken over all possible pairs of complementary observables. 

\begin{remark}
When the maximum number of complementary observables is not known, $v$ is substituted by the maximum known cardinality of the set of complementary observables, and the above expression becomes a lower bound of the total correlations. 
\end{remark}

It can be easily seen that for any bipartite state the total correlations are at least as large as the maximum correlations of a single pair of observables, i.e., $\pcc_{\qs}(\text{A} : \text{B}) \geq \max_{\{ \m{A} , \m{B} \}}  |\cor_{\qs}(\m{A} , \m{B})| $. In the section below we focus on two-dimensional bipartite states, i.e., two-qubit states, and show what is the implication when this inequality becomes an equality.

\section{Total correlations in two-Qubit States}
\label{sec:total_correlations_qubits}

Let us focus in this section on two-qubit states and two-dimensional observables. Given that the absolute PCC is invariant under affine transformations~\cite{Tserkis_etal_PLA_24}, without loss of generality we can always consider observables that have dichotomic outcomes: $\pm1$. Complementary observables with such eigenvalues are always orthogonal to each other~\cite{Tserkis_etal_PLA_24}, so any triplet of them constitutes a Hermitian operator basis. 

\begin{theorem}
A necessary condition for a two-qubit state to be classically correlated is the ability of its total correlations to be concentrated using local unitary operations in a single pair of observables, i.e.,
\eq{}{
\pcc_{\qs}(\text{A} : \text{B}) = \max_{\{ \m{A} , \m{B} \}}  \left| \cor_{\qs}(\m{A} , \m{B}) \right| \,.
}
\end{theorem}

\begin{proof}
Every classically correlated two-qubit state [see Eq.~\eqref{eq:classically_correlated}] can be diagonalized through local unitary operations (which leave the total correlations invariant), so it can take the following form:
\eq{Fano_cl}{
\qs_{\text{cl}}= \frac{1}{4} \big( \id{} \otimes \id{} + n_3 \m{\Lambda}_3 \otimes \id{} + s_3 \id{} \otimes \m{\Lambda}_3 + t_{33}\m{\Lambda}_3 \otimes \m{\Lambda}_3 \big) \,,
}
where $\m{\Lambda}_3 = \ketbra{0}{0} - \ketbra{1}{1}$ \,.

The PCC for two arbitrary observables $(\m{A}, \m{B})$ is equal to
\eq{}{ 
\cor_{\qs_{\text{cl}}}(\m{A}, \m{B}) = \frac{\frac{t_{33}-n_3 s_3}{4} \trc(\m{A} \m{\Lambda}_3 ) \trc( \m{B} \m{\Lambda}_3 ) }{\sqrt{\left[ 1 - \frac{n_3^2}{4} \trc(\m{A} \m{\Lambda}_3 )^2 \right] \left[ 1 - \frac{s_3^2}{4} \trc(\m{B} \m{\Lambda}_3 )^2 \right] }} \,,
}
since $\expv_{\qs_{\text{cl}}}(\m{A} \otimes \id{}) = \frac{1}{2} n_3 \trc \left( \m{A} \m{\Lambda}_3 \right)$, $\expv_{\qs_{\text{cl}}}(\id{} \otimes \m{B}) = \frac{1}{2} s_3 \trc \left( \m{B} \m{\Lambda}_3 \right)$, and $\expv_{\qs_{\text{cl}}}(\m{A} \otimes \m{B}) = \frac{1}{4} t_{33}\trc \left( \m{A} \m{\Lambda}_3 \right) \trc \left( \m{B} \m{\Lambda}_3 \right)$. Given that $\trc(\m{A} \m{\Lambda}_3 ) \in [-2,2]$ and $\trc(\m{B} \m{\Lambda}_3 ) \in [-2,2]$, we have
\eq{}{ 
\max_{\{ \m{A} , \m{B} \}}  \left| \cor_{\qs_{\text{cl}}}(\m{A} , \m{B}) \right| = \frac{|t_{33}-n_3 s_3| }{\sqrt{\left( 1 - n_3^2 \right) \left( 1 - s_3^2 \right) }} \,,
}
that is reachable for $\m{A} = \m{B}=\m{\Lambda}_3$.

For the sum of correlations we have
\alsub{}{ 
&\sum_{i=1}^{3}  \frac{|t_{33} {-}n_3 s_3| |\trc(\m{A}_i \m{\Lambda}_3 ) \trc( \m{B}_i \m{\Lambda}_3 )| }{4\sqrt{\left[1 {-} \frac{n_3^2}{4} \trc(\m{A}_i \m{\Lambda}_3 )^2 \right] \left[ 1 {-} \frac{s_3^2}{4} \trc(\m{B}_i \m{\Lambda}_3 )^2 \right]}}  \\
& \leq \frac{|t_{33}-n_3 s_3|}{4\sqrt{(1-n_3^2)(1-s_3^2)}} \sum_{i=1}^{3} |\trc(\m{A}_i \m{\Lambda}_3 ) \trc( \m{B}_i \m{\Lambda}_3 )| \,. 
}
Using the Cauchy-Schwarz inequality, we obtain
\eq{}{
\sum_{i=1}^{3} |\trc(\m{A}_i \m{\Lambda}_3 ) \trc( \m{B}_i \m{\Lambda}_3 )| \leq \sqrt{ \sum_{i=1}^{3} \trc(\m{A}_i \m{\Lambda}_3 )^2  \sum_{i=1}^{3}\trc( \m{B}_i \m{\Lambda}_3 )^2} \,,
}
which, due to the Parseval’s identity, is simplified to 
\eq{}{
\sum_{i=1}^{3} |\trc(\m{A}_i \m{\Lambda}_3 ) \trc( \m{B}_i \m{\Lambda}_3 )| \leq 4 \,.
}
Thus, the maximum sum of correlations is equal to
\eq{}{ 
\pcc_{\qs_{\text{cl}}}(\text{A} : \text{B}) = \frac{|t_{33}-n_3 s_3|}{\sqrt{(1-n_3^2)(1-s_3^2)}} \,,
}
that is reachable for $\m{A}_i=\m{B}_i=\m{\Lambda}_i \,\, \forall i \in \{ 1,2,3\}$, which completes the proof.
\end{proof}

The above result implies the following sufficient condition for a two-qubit state to be quantumly correlated.

\begin{corollary}
A sufficient condition for a two-qubit state to be quantumly correlated is its total correlations to be over one.
\end{corollary}

\begin{proof}
Based on Theorem 1 any classically correlated two-qubit state can concentrate its total correlations into a single pair of observables.  Given that the correlations of every single pair cannot exceed one it is easy to see that if the total correlations of a state exceed one then the state must be quantumly correlated.
\end{proof}

Let us consider now two-qubit states written in the standard form~\cite{Leinaas_Myrheim_Ovrum_PRA_06}, i.e.,
\eq{standard_form}{
\qs_{\text{sf}}= \frac{1}{4} \bigg( \id{} \otimes \id{} + \sum_{k=1}^{3} t_k \m{\Lambda}_k \otimes \m{\Lambda}_k \bigg) \,
}
Those are states that become maximally mixed when one of the two parts is traced out, and are also known as Bell-diagonal states, since they can be written as a mixture of the four Bell states. 

\begin{remark}
A two-qubit state in the standard form is classically correlated if and only if it has only one non-zero coefficient $t_k$~\cite{Dakic_Vedral_Brukner_PRL_10}. 
\end{remark}

\begin{lemma}
For a two-qubit state in the standard form, $\qs_{\text{sf}}$, the PCC is given by
\eq{stform}{
\cor_{\qs_{\text{sf}}}(\m{\Lambda}_i,\m{\Lambda}_i) = t_i \quad \forall i \in \{1, 2, 3 \} \,.
}
\end{lemma}

\begin{proof}
Simple calculations lead us to $\expv_{\qs_{\text{sf}}}(\m{\Lambda}_i) = 0$ and $\vari_{\qs_{\text{sf}}}(\m{\Lambda}_i) = 1 $ for observables belonging in each partition. We also have $\expv_{\qs_{\text{sf}}}(\m{\Lambda}_i \otimes\m{\Lambda}_i) = t_i \,\, \forall i \in \{1, 2, 3 \}$. Combining the above results we arrive at Eq.~\eqref{eq:stform}, which completes the proof.
\end{proof}

\begin{theorem}
A necessary and sufficient condition for a two-qubit state in the standard form to be classically correlated is the ability of its total correlations to be concentrated using local unitary operations in a pair of observables, i.e.,
\eq{}{
\pcc_{\qs_{\text{sf}}}(\text{A} : \text{B}) = \max_{\{ \m{A} , \m{B} \}}  \left| \cor_{\qs_{\text{sf}}}(\m{A} , \m{B}) \right| \,.
}
\end{theorem}

\begin{proof}
The necessity of this condition directly follows from Theorem 1, so we only need to prove the sufficiency of it. Consider two observables decomposed in the Pauli basis as follows: $\m{A}=a_0 \id{} + \sum_{k=1}^3 a_k \, \m{\Lambda}_k$ and $\m{B}=b_0 \id{} + \sum_{k=1}^3 b_k \, \m{\Lambda}_k$. Then, based on ~\cite{Tserkis_etal_PLA_24} the PCC can be expressed as
\eq{}{
\cor_{\qs}(\m{A} , \m{B}) = \frac{\braf{\check{a}} \m{C} \ketf{\check{b}}}{\sqrt{1- \braketf{\check{a}}{n}^2}\sqrt{1-\braketf{\check{b}}{s}^2}} \,,
}
where $\ketf{\check{a}} = [\check{a}_1, \check{a}_2,  \check{a}_3]^T$ and $\ketf{\check{b}} = [\check{b}_1, \check{b}_2, \check{b}_3]^T$ are two normalized vectors with $\check{a}_k = a_k / \sqrt{a_1 + a_2 + a_3}$ and $\check{b}_k = b_k / \sqrt{b_1 + b_2 +b_3}$. $\m{C}$ is a matrix with elements $[\m{C}]_{k \ell}=t_{k \ell} - n_k s_{\ell}$, and $\ket{n}$ and $\ket{s}$ the vectors $[n_1, n_2, n_3]^T$ and $[s_1, s_2, s_3]^T$, respectively [see Eqs.~\eqref{eq:n_s_t_coeff}]. Then, we have
\eq{}{
\cor_{\qs}(\m{A} , \m{B}) \leq \frac{\braf{\check{a}} \m{C} \ketf{\check{b}}}{\sqrt{1 - \braket{n}{n}} \sqrt{1 - \braket{s}{s} } } \,.
}
since $\braket{\check{a}}{n}^2 \leqslant \sum_{i=1}^3 \braket{\check{a}_i}{n}^2 = \sum_{i=1}^3 n_i^2 = \braket{n}{n}$, and analogously $\braketf{\check{b}}{s}^2 \leqslant \braket{s}{s}$. The spectral norm of $\m{C}$, $\norm{\m{C}}_2$, can be expressed as $\norm{\m{C}}_2 = \max\limits_{\bigl\{ \ketf{\check{a}} , \ketf{\check{b}} \bigl\}} |\braf{\check{a}} \m{C} \ketf{\check{b}}|$~\cite{Horn_Johnson_B_13}, so we obtain
\eq{spectral_norm}{
\max_{\{ \m{A} , \m{B} \}}  |\cor_{\qs}(\m{A} , \m{B})| \leq \frac{ \norm{\m{C}}_2}{\sqrt{1 - \braket{n}{n}} \sqrt{1 - \braket{s}{s}}} \,,
}
where $\norm{\m{C}}_2$ is equal to the maximum singular value of $\m{C}$. It was also shown in Ref.~\cite{Tserkis_etal_PLA_24}, that for two-qubit separable states we have
\eq{trace_norm}{
\pcc_{\qs}(\m{A}_{i} : \m{B}_{i} ) \leq \frac{ \norm{\m{C}}_{\text{tr}}}{\sqrt{1 - \braket{n}{n}} \sqrt{1 - \braket{s}{s}}} 
}
where $\norm{\m{C}}_{\text{tr}}$ denotes the trace norm of $\m{C}$, i.e., the sum of its singular values~\cite{Horn_Johnson_B_13}. 

For two-qubit states in the standard form, inequalities \eqref{eq:spectral_norm} and \eqref{eq:trace_norm} become:
\eq{}{
\max_{\{ \m{A} , \m{B} \}}  |\cor_{\qs_{\text{sf}}}(\m{A} , \m{B})| \leq \max\{|t_1|, |t_2|, |t_3|\}
}
and
\eq{}{
\pcc_{\qs}(\m{A}_{i} : \m{B}_{i} ) \leq \sum_{i=1}^3 |t_i| \,,
}
whose upper bounds are reachable according to Lemma 1 for $\m{A}_i = \m{B}_i = \m{\Lambda}_i$.

The condition $ \pcc_{\qs_{\text{sf}}}(\m{A}_{i} : \m{B}_{i} ) = \max_{\{ \m{A} , \m{B} \}}  |\cor_{\qs_{\text{sf}}}(\m{A} , \m{B})|$ is equivalent to the condition $\max\{|t_1|, |t_2|, |t_3|\} = \sum_{i=1}^3 |t_i|$, that is true only when the state is classically correlated (see Remark 4), which completes the proof.
\end{proof}

\begin{figure*}[t]
\centering
\includegraphics[width=\textwidth]{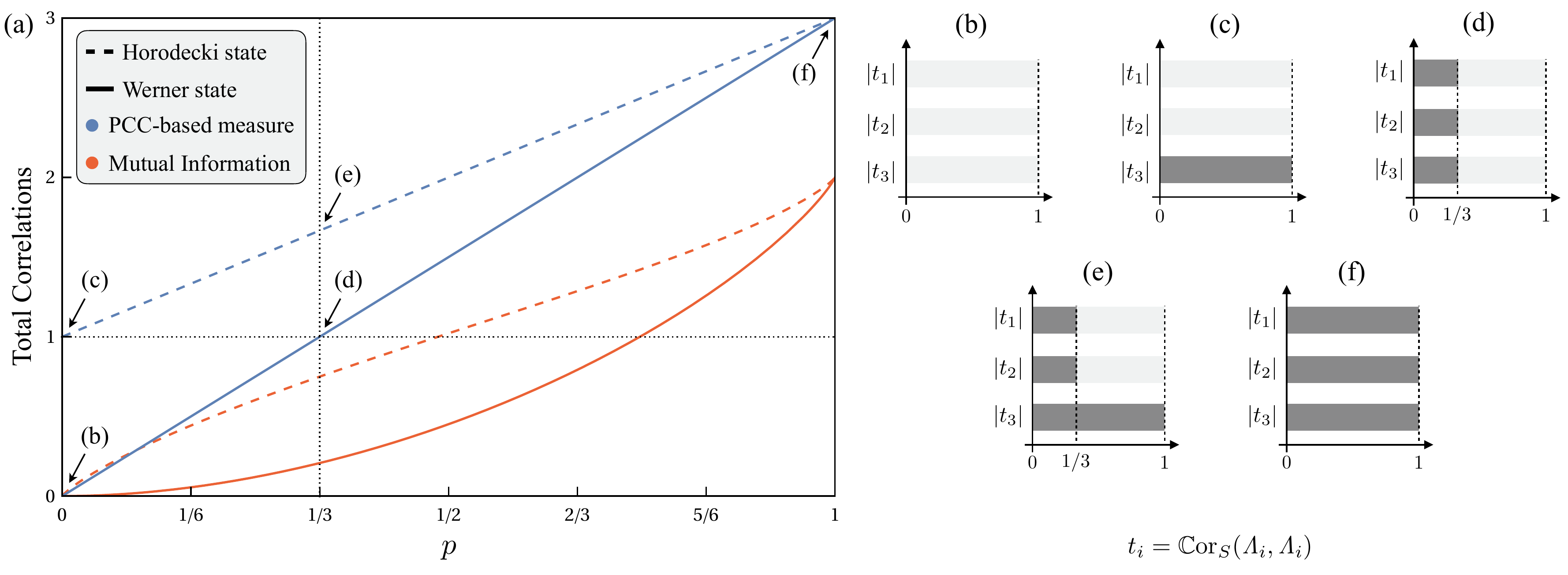}
\caption{In panel (a) the total correlations of Werner states $\qs_{\text{W}}$ and Horodecki states $\qs_{\text{H}}$ [see Eqs.~\eqref{eq:werner_state} and \eqref{eq:horodecki_state}] are quantified by the PCC-based measure $\pcc_{\qs}(\text{A} : \text{B})$ and the mutual information $\mi_{\qs}(\text{A}:\text{B})$. In panels (b)-(f) specific points are isolated and analyzed in terms of the PCC of each observable pair. It is clear that the PCC provides in general more granularity indicating how the correlations are distributed among the pairs of observables in the system.}
\label{fig:main_plot}
\end{figure*}

\section{Discussion}
\label{sec:discussion}

In Fig.~\ref{fig:main_plot} we compare the two ways total correlations can be quantified: through the expression $\pcc_{\qs}(\text{A} : \text{B})$ and through mutual information $\mi_{\qs}(\text{A}:\text{B})$. We pick two states for this comparison: the state
\eq{werner_state}{
\qs_{\text{W}} \coloneqq p \m{\Phi} + \frac{1-p}{4} \id{} \,,
}
namely the Werner state, where $\m{\Phi} = \ketbra{\phi}{\phi}$ is the Bell state with $\ket{\phi} = (\ket{00} + \ket{11})/\sqrt{2}$, and the state
\eq{horodecki_state}{
\qs_{\text{H}} \coloneqq p \m{\Phi} + (1-p) \m{\Pi}_0 \,,
} 
which is a variation of the so-called Horodecki state, where $\m{\Pi}_0 = \ketbra{00}{00}$. In panel (a) we plot the total correlations for both states through both methods. Note that the three optimal pairs of observables in Eq.~\eqref{eq:total_correlations_measure} for both states are: $( \m{\Lambda_1, \m{\Lambda}_1} )$, $ ( \m{\Lambda_2, \m{\Lambda}_2})$, and $( \m{\Lambda_3, \m{\Lambda}_3})$.

When $p = 1$ both states are in the Bell state $\m{\Phi}$, so both expressions reach their maximum value. In panel (f) we focus on the value $p=1$, for which the PCC of all three pairs reaches unity. For $1/3 < p \leq 1$, $\qs_{\text{W}}$ is entangled, and focusing in particular in the case for $p=1/3$ we see in panel (d) that the PCC's of all three pairs are equal to 1/3, summing up to 1, which is the separability threshold~\cite{Tserkis_etal_PLA_24}. On the other hand, the corresponding value $\mi_{\qs_{\text{W}}}(\text{A}:\text{B}) \simeq 0.21$ for $p=1/3$ does not constitute a separability bound. $\qs_{\text{H}}$ is entangled for any $p>0$, which is also reflected in panel (a) since $\pcc_{\qs_{\text{H}}}(\text{A} : \text{B})$ is always above the unity threshold. In panel (e) we see that $\qs_{\text{H}}$ for $p=1/3$ has perfect correlations for one of the three observable pairs while for the other two the correlations are equal to 1/3. In that case $\mi_{\qs_{\text{H}}}(\text{A}:\text{B}) \simeq 0.75$. When $p=0$, $\qs_{\text{W}}$ becomes the maximally mixed state, and both $\pcc_{\qs_{\text{W}}}(\text{A} : \text{B})$ and $\mi_{\qs_{\text{W}}}(\text{A}:\text{B})$ become equal to zero. In panel (b) we see how all PCC's vanish for $p=0$ for $\qs_{\text{W}}$. Interestingly, for $\qs_{\text{H}}$ when $p \rightarrow 0$, i.e., when $\qs_{\text{H}}$ approaches a pure product state, we have $\pcc_{\qs_{\text{H}}}(\text{A} : \text{B}) \rightarrow1$ and $\mi_{\qs_{\text{H}}}(\text{A}:\text{B}) \rightarrow 0$. This example shows why both quantification methods are valuable since they provide a different perspective on the quantification of correlations.

As we discussed in Remark 2, assessing correlations through distance-based measures leads to an issue, which can be avoided all together if we abandon the assumption that quantum states can simultaneously be classically and quantumly correlated. Using our method the distinction between classical and quantum correlated states depends on whether multiple pairs of observables are simultaneously correlated (recall that the pairs of observables are not arbitrary but within each partition the considered observables are mutually incompatible as seen in Fig.~\ref{fig:twosides}). For example using $\pcc$ the states in panels (c) and (d) have the same total correlations, however, the state in panel (c) is a classically correlated state because only one pair of observables is correlated, but the state in panel (d) is a quantumly correlated state since all three pairs are correlated. 

More broadly, our analysis suggests that different types of correlations appear as different distributions of the PCC values between pairs of observables. For two-qubit states in the standard form of Eq.~\eqref{eq:standard_form} we were able to show that the maximum total correlations can be concentrated in a single pair of observables if and only if the state is classically correlated. Numerical analysis provides us evidence that the same condition should hold for any two-qubit quantum state, even though a rigorous proof is currently missing. For higher-dimensional states it is yet unclear how exactly the distribution of the PCC values is related to the nature of the correlations in the system, and more work is needed in this direction. 

It is now clear how classical and quantum correlations are connected to the uncertainty principle, i.e., complementarity, and why entanglement is commonly referred to as a ``stronger type of correlation'' in comparison to the classical ones. Perfect correlations between two observables can exist in both classical and quantum systems. However, at least for states in the form of Eq.~\eqref{eq:standard_form} the fundamental difference between classical and quantum correlations is that, in the former, correlations can be concentrated within a single pair of complementary observables [e.g., see Fig.~\ref{fig:main_plot}(c)], whereas this is not possible in the latter. The concentration of correlations into a single pair of observables can be understood as follows: in classical systems---where only classical correlations can be present---properties are compatible to each other, so the maximum number of complementary observables per system is trivially equal to one. Thus, in classical systems the existence of multiple pairs of complementary observables is fundamentally prohibited, a restriction that does not apply to quantum systems.

\section{Conclusions}
\label{sec:conclusion}

In this work, we introduced a framework for the quantification of total correlations in quantum systems through the Pearson correlation coefficient. For two-qubit states in particular, we showed how quantum correlations are connected to the uncertainty principle, and how classically and quantumly correlated states can be distinguished based on the correlations between two specific sets of observables. These sets correspond to locally complementary observables, which are also useful for other tasks, such as quantum state estimation~\cite{Adamson_Steinberg_PRL_10, Fernandez_Saavedra_PRA_11, Petz_Laszlo_RMP_12}, entanglement detection~\cite{Spengler_etal_PRA_12, Maccone_Bruss_Macchiavello_PRL_15}, quantum cryptography~\cite{Bennett_Brassard_IEEE_84, Cerf_etal_PRL_02}, quantum retrodiction~\cite{Englert_Aharonov_PLA_01, Reimpell_Werner_PRA_07}. An interesting future path of this work would be to show that the condition in Theorem 1 is not only necessary but also sufficient, and potentially extend this result to higher-dimensional and/or multi-partite quantum states.

\section*{Acknowledgments}
This work is supported by the National Science Foundation under grant number NSF CNS 2106887 on ``U.S.-Ireland R\&D Partnership: Collaborative Research: CNS Core: Medium: A unified framework for the emulation of classical and quantum physical layer networks'', the NSF QuIC-TAQS program ``QuIC-TAQS: Deterministically Placed Nuclear Spin Quantum Memories for Entanglement Distribution'' under grant number NSF OMA 2137828, the NSF CAREER: First Principles Design of Error-Corrected Solid-State Quantum Repeaters under grant number \#1944085, the Australian Research Council Centre of Excellence CE170100012, Laureate Fellowship FL150100019, and the Australian Government Research Training Program Scholarship.

\bibliography{Bibliography/bibliography}

\end{document}